\documentclass[11pt]{article}
\usepackage{geometry}
\usepackage{amssymb}
\usepackage{mathrsfs}
\usepackage{amsmath}

\newtheorem{theorem}{Theorem}[section]

\newtheorem{lemma}[theorem]{Lemma}
\newtheorem{proposition}[theorem]{Proposition}
\newtheorem{claim}[theorem]{Claim}
\newtheorem{definition}[theorem]{Definition}

\def\squarebox#1{\hbox to #1{\hfill\vbox to #1{\vfill}}}
\newcommand{\qed}{\hspace*{\fill}
\vbox{\hrule\hbox{\vrule\squarebox{.667em}\vrule}\hrule}\smallskip}

\newenvironment{proof}{\noindent{\bf Proof:~~}}{\(\qed\)}

\newcommand{\ignore}[1]{}

\begin{document}
\title{Truthfulness via Proxies}

\author{Shahar Dobzinski \\
  Department of Computer Science\\
  Cornell Unversity\\
  {\tt shahar@cs.cornell.edu}
  \and Hu Fu \\
  Department of Computer Science\\
  Cornell Unversity\\
  {\tt hufu@cs.cornell.edu}
  \and Robert Kleinberg\\
  Department of Computer Science\\
  Cornell Unversity\\
  {\tt rdk@cs.cornell.edu}}
\maketitle

\begin{abstract}
This short note exhibits a truthful-in-expectation $O(\frac {\log m} {\log \log m})$-approximation mechanism for combinatorial auctions with subadditive bidders that uses polynomial communication.
\end{abstract}

\section{Introduction}

We consider the problem of maximizing the social welfare in combinatorial auctions with subadditive bidders. In this problem, we have a set $M$, $|M|=m$, of heterogeneous items, and $n$ bidders. Each bidder $i$ has a valuation function $v_i$, $v_i:2^M\rightarrow \mathbb R$. We assume that the each valuation $v_i$ is normalized ($v_i(\emptyset)=0$), non-decreasing, and subadditive: for each two bundles $S$ and $T$, $v_i(S)+v_i(T)\geq v_i(S\cup T)$. This important class of valuations contains other interesting classes. In particular, it strictly contains the class of submodular valuations. Our goal is to find an allocation of the items $(S_1,\ldots, S_n)$ that maximizes the welfare: $\Sigma_iv_i(S_i)$. The efficiency of our algorithms will be measured in terms of the natural parameters of the problem: $n$ and $m$. Since the valuation functions are objects of exponential size, we assume that they are represented by black boxes that can answer any type of queries. In particular, we assume that bidders are computationally unbounded and measure only the amount of communication between them.

Feige \cite{F06} obtains a $2$-approximation algorithm by applying an ingenious randomized rounding. This is the best possible \cite{DNS05} in polynomial communication. Much research was concerned with the design of truthful algorithms for the problem: algorithms in which a profit-maximizing strategy for each bidder is to truthfully answer the queries. The simplest and strongest notion considered is deterministic truthfulness, when no randomization is allowed. The VCG mechanism is truthful and optimally solves the problem, but it is not computationally efficient. The best truthful approximation algorithm known achieves a poor approximation ratio of $O(\sqrt m)$ \cite{DNS05}, which is achieved by a maximal-in-range algorithm (see below). An evidence that deterministic algorithms cannot achieve an improved approximation ratio was given in \cite{DN07a}: maximal-in-range algorithms cannot achieve an approximation ratio better than $m^{\frac 1 6}$ with polynomial communication.

If randomization is allowed, there exists a universally truthful mechanism (a distribution over deterministic truthful mechanisms) that guarantees an approximation ratio of $O(\log m\log \log m)$ \cite{D07}. This paper relaxes the solution concept and considers mechanisms that are \emph{truthful in expectation}: truth telling maximizes the \emph{expected} profit of each bidder, where the expectation is taken over the internal random coins of the algorithm. We stress that truthful-in-expectation is much weaker than universal truthfulness: in particular, bidders that are not risk neutral may not be incentivized to bid truthfully. See \cite{DNS06} for a discussion.

In this note we show that there exists a truthful-in-expectation $O(\frac {\log m} {\log \log m})$-approximation mechanism that uses only polynomial amount of communication. While this only slightly improves the best approximation ratio provided by universally truthful algorithms (at the cost of weakening the solution concept), we feel that this result is of interest for two main reasons (1) the random-sampling based techniques of \cite{D07} do not seem capable of achieving a better than a logarithmic factor, so this result might hint at a performance gap between truthful-in-expectation and universally truthful mechanisms in our setting, and (2) our rounding technique is non-trivial and may be of independent interest.

Our algorithm is described in the communication complexity model. The model is best suited for proving strong, unconditional impossibility results. In particular, for algorithms it should not take advantage of the unlimited computational power of the bidders. Indeed, essentially all known algorithms known for combinatorial auctions need polynomially bounded computation power and use a restricted and natural communication form, e.g., demand queries (given prices per item $p_1,\ldots ,p_m$ return a bundle that maximizes the profit). Unfortunately, we do not know how to implement our algorithm using demand queries, or any other type of natural query. In a sense, this is an abuse of the communication complexity model; hence we recommend viewing our result as demonstrating the falsehood of certain strong impossibility results for the problem, rather than as a ``real'' mechanism suitable for use in combinatorial auctions.

Let us briefly discuss the techniques we use. Almost all truthful deterministic mechanisms known in the literature are \emph{maximal-in-range}, a scaled-down version of the classic VCG mechanism: for every set of valuations, select the welfare-maximizing allocation in a \emph{predefined} set of allocations. However, as mentioned above, \cite{DN07a} shows that this method can only guarantee a poor approximation ratio in our setting. One way to overcome this obstacle was spelled out by \cite{DD09} which suggested to use \emph{maximal in distributional range} (MIDR) mechanisms: for every set of valuations, select the welfare-maximizing distribution in a preselected set of \emph{distributions over allocations}, then sample an allocation from this distribution. This results in a truthful-in-expectation mechanism. Furthermore, in \cite{DD09} it is shown that for some problems MIDR mechanisms may be more powerful than any universally truthful mechanism.

Our mechanism is also MIDR\footnote{In fact, our range will consist of distributions over \emph{infeasible} allocations, but we will always output a distribution over feasible allocations with the same expected welfare as the best distribution in the range. Hence the mechanism is \emph{equivalent} to MIDR (see \cite{DN07a}), and truthfulness in expectation follows.}. The basic idea is to represent each bidder by a proxy bidder, find the optimal fractional solution among these proxy bidders, and then round the fractional solution. Each proxy bidder is defined so that he ``simulates'' the expected value of the bundle after the randomized rounding. The main obstacle is to prove feasibility while still being able to relate the value of the rounded solution (that was calculated with respect to the proxy bidders) to the original bidders.

We note that Lavi and Swamy \cite{LS05} already implicitly used maximal-in-distributional-range algorithms together with the LP relaxation of the problem. Our solution requires a subtler rounding of the LP, and in particular overcomes one of the main barriers of \cite{LS05}: their decomposition is based on an algorithm that ``verifies'' the integrality gap, and thus they can only provide a truthful-in-expectation $O(\sqrt m)$-approximation mechanism for our setting. Our mechanism uses a direct approach to ``decompose'' a linear program with proxy bidders. This is one of the main reasons for our success in guaranteeing a better approximation ratio.

The main open question we leave open is the existence of a constant-ratio truthful-in-expectation mechanism for combinatorial auctions with subadditive bidders. A first step in this direction might be to prove that no polynomial time MIDR mechanism can achieve an $O(1)$ approximation in polynomial time.

\section{The Mechanism}

\subsection{The Range}

We first remind the reader the linear program relaxation of the problem:

\noindent  \emph{Maximize:}
  $\Sigma_{i,S}x_{i,S}v_i(S)$

\noindent  \emph{Subject to:}
  \begin{itemize}
      \item For each item $j$: $\Sigma_{i,S|j\in S}x_{i,S}\leq 1$
      \item for each bidder $i$: $\Sigma_{S}x_{i,S}\leq 1$
      \item for each $i$, $S$: $x_{i,S}\geq 0$
  \end{itemize}

\begin{definition}
A distribution $D$ over (not necessarily feasible) allocations is called \emph{$(c,p)$-fractional} if there exists a feasible fractional solution to the LP such that $D$ is equal to the distribution produced by the following process: with probability $p$ no bidder is allocated any item. With probability $1-p$ each bidder $i$ receives exactly one bundle $S$ with probability $x_{i,S}$ and keeps each item $j\in S$ with probability $c$, independently at random.
\end{definition}

\begin{definition}[the range]
The range $\mathcal R_{c,p}$ consists of all $(c,p)$-fractional distributions.
\end{definition}

Before proving that there is one distribution in $\mathcal R_{c,p}$ that provides a good approximation, we require the following definition:

\begin{definition}
Given a valuation $v$, let the \emph{$c$-proxy valuation} $v'$ be:
\[
v'(S)=E_{T\sim\mathcal D_S}[v(T)]
\]
where $\mathcal D$ is the distribution where each $j\in S$ is in $T$ with probability $c$ independently at random.
\end{definition}

\begin{lemma}
The optimal distribution in a $\mathcal R_{c,p}$ range is an $O(\frac {1} {c\cdot p})$-approximation to the optimal social welfare if the valuations are subadditive.
\end{lemma}

The proof of this lemma immediately follows by considering the optimal allocation and using the following proposition from \cite{F06}:
\begin{proposition}[paraphrasing \cite{F06}]
Let $S$ be a bundle, $v$ a subadditive valuation, and $c$ such that $0\leq c\leq 1$ and $\frac 1 c$ is an integer. Then, $v'(S)\geq c\cdot v(S)$, where $v'$ is the $c$-proxy valuation of $v$.
\end{proposition}

\subsection{The Algorithm}

\begin{enumerate}
\item For each valuation $v_i$ of bidder $i$, let $v'_i$ be the $c$-proxy valuation.
\item\label{step-LP} Solve the linear program relaxation of the problem with respect to the \emph{$c$-proxy valuations} $v'_i$.

\item\label{step-tentative} Each bidder $i$ is tentatively assigned exactly one bundle $S_i$, where bundle $S$ is allocated to $i$ with probability exactly $x_{i,S}$. If there is an item that is allocated more than $\frac 1 c$ times, the algorithm halts and no bidder is allocated any items. Otherwise, proceed to the next steps. 
    
\item For each bidder $i$, let $q_i$ be the probability that some item in $S_i$ was allocated more than $\frac 1 c-1$ times in the following random experiment: each bidder $i'$, $i'\neq i$, is allocated bundle $S$ with probability $x_{i,S}$.

\item\label{step-item-lottery} Independently for each item $j\in S$, select at most one bidder to receive $j$, so that each bidder that is tentatively allocated $S_i$ such that $j\in S_i$ receives item $j$ with probability exactly $c$.

\item\label{step-personal-cancel} For each bidder $i$, with probability $1-\frac p {1-q_i}$ he is not allocated any items, and with probability $\frac p {1-q_i}$ he keeps the items that he was assigned in the previous step.
\end{enumerate}

We note that the last step of the algorithm is a re-implementation of the main idea behind the truthful-in-expectation mechanism for single-minded bidders of \cite{APTT03}.

In our proofs we assume $m$ is large enough (when $m$ is a constant VCG can be implemented in polynomial time in $n$). Also, from now on we fix and $c=\frac {\log \log m} {100 \log m}$ and $p=\frac 1 {20}$. We first have to make sure that the algorithm is correctly defined. For that we have to prove that $p\geq 1-q_i$, for every $q_i$. This follows from the next claim \cite{F06,DNS05}:
\begin{claim}[\cite{F06}]
Fix any feasible solution of the linear program. Allocate each bidder $i$ exactly one bundle where each bundle $S$ is allocated with probability $x_{i,S}$. The probability that no item is allocated more than $\frac 1 c$ times is at least $1-\frac 1 m$.
\end{claim}
In other words, $q_i\leq \frac 1 m\leq p$, as needed. We now show that the algorithm indeed finds the optimal distribution in the range and uses a polynomial amount of communication.

\begin{lemma}
The algorithm finds a distribution with value that equals the distribution with the maximum expected welfare in $\mathcal R_{(c,p)}$.
\end{lemma}
\begin{proof}
Notice that the optimal solution of the linear program with the proxy valuations is exactly the expected value of the optimal distribution in  $\mathcal R_{c,1}$: the proxy valuations ``simulate'' the random process where each bidder keeps each item with probability $c$. Thus, it suffices to prove that the algorithm always find a solution with value \emph{exactly} $p\cdot OPT$, where $OPT$ is the value of the optimal solution of the LP. This will guarantee us that the expected value of the solution equals $\mathcal R_{(c,p)}$.

Notice that after Step \ref{step-tentative} the expected value of the sum of the tentative bundles (with respect to the proxy valuations) is exactly the value of the optimal solution of the linear program. After step \ref{step-item-lottery} the expected value of the bundles assigned to the bidders (now with respect to the real valuations) is greater than $p\cdot OPT$. The last step ``cancels'' some of the allocations so that if bidder $i$ was allocated bundle $S$, the probability he will be allocated some items from $S$ (i.e., the probability that auction is not canceled at Step \ref{step-tentative} and that he keeps some items at Step \ref{step-personal-cancel}) is \emph{exactly} $p\cdot x_{i,S}$. Thus the expected value of the solution is exactly $p\cdot OPT$, as needed.
\end{proof}

\begin{lemma}
The communication complexity of the algorithm is polynomial in $n$ and $m$.
\end{lemma}
\begin{proof}
The only two steps for which it is not obvious that only polynomial communication is required are Steps \ref{step-LP} and the calculation of the $q_i$'s. In Step \ref{step-LP} we solve a linear program that calculates the optimal fractional solution for some valuations. This can be done in polynomial communication, as long as demand queries are available \cite{BN09}. In our case we need to compute answers to demand queries with respect to the proxy valuations. This can be done by each bidder $i$ with no additional communication since each proxy valuation $v'_i$ depends only on $v_i$. We note that the support of the solution of the linear program consists only of polynomially many variables \cite{BN09}.

Calculating each $q_i$ can be done using no communication cost simply by enumerating over all possible outputs of the random coins.
\end{proof}

Together we have:
\begin{theorem}
There exists a truthful-in-expectation $O(\frac {\log m} {\log \log m})$-approximation mechanism for combinatorial auctions with submodular bidders. The algorithm uses polynomial communication.
\end{theorem}

%
%\section{The Framework}
%
%\begin{definition}
%A function $f:V\rightarrow V$ is called a \emph{proxy}.
%\end{definition}
%
%\begin{proposition}
%Let $f_1,\ldots ,f_n$ be known proxies for bidders $1,\ldots ,n$. Suppose that for every $v_1,\ldots ,v_n$ we have a (possibly randomized) algorithm $A$ for $f_1(v_1),\ldots,f_n(v_n)$ such that:
%\begin{itemize}
%\item $A$ is truthful.
%\item Let $(S_1,\ldots, S_n)$ be the output of $A$. Then, there is a randomized
%\end{itemize}
%\end{proposition}

\subsubsection*{Acknowledgments}
We thank Shaddin Dughmi and Tim Roughgarden for pointing out a bug in an earlier version of this note.

\bibliographystyle{plain}
\bibliography{../bib}

\end{document}